\documentclass[twoside]{article}
\usepackage{amsmath,amssymb,amsthm,latexsym,bm,indentfirst,wasysym}
\usepackage{graphicx,epsfig,amscd,mathrsfs,multirow,bm}
\usepackage{cite}
\usepackage{caption}
\usepackage{color}

\usepackage{booktabs}
\usepackage{algorithm}
\usepackage{algorithmic}

\usepackage{caption}
\usepackage[colorlinks,linkcolor=red,anchorcolor=red,citecolor=blue]{hyperref}
\numberwithin{equation}{section}
\usepackage{titlesec}
\titleformat{\section}{\large\bfseries}{}{0pt}{}

\input xy
\xyoption{all}
\input scrload.tex

\allowdisplaybreaks[4]

\renewcommand{\arraystretch}{1.2}
\catcode`@=11
\long\def\@makefntext#1{\noindent #1}
\newskip\tabcentering \tabcentering=1000pt plus 1000pt minus 1000pt
\def\MCH#1#2{\setbox0=\hbox{\raise#1\hbox{#2}}\smash{\box0}}% move char

\def\@evenfoot{}\def\@oddfoot{}

\def\@evenhead{\hbox to\textwidth{\small\rm\thepage \hfill
{\it  Xingfu Li}}} % authors name (the given name is before the surname, and use "and" to separate two authors)%

\def\@oddhead{\hbox to \textwidth{\small{\it
Complexity of co-linear chaining
} \hfill\thepage}}   % An abbreviated title not exceeding 45 characters (including spaces) for the running head

\catcode`@=12

\def\bc{\begin{center}}
\def\ec{\end{center}}
\def\no{\noindent}
\def\hang{\hangindent\parindent}
\def\textindent#1{\indent\llap{\qquad #1\ \ \enspace}\ignorespaces}
\def\ref{\par\hang\textindent}

\begin{document}

 \newtheorem{theorem}{Theorem}
 \newtheorem{lemma}{Lemma}
 %\newtheorem{proof}{Proof}
 %\newproof{proof}{Proof}

%\linenumbers

\abovedisplayskip=6pt plus 1pt minus 1pt \belowdisplayskip=6pt
plus 1pt minus 1pt
%-------------------  First Head  -----------------------------------------
\thispagestyle{empty} \vspace*{-1.0truecm} \noindent

%===================Text=============================================
\vskip 10mm \bc{\Large\bf %Sequence-to-graph alignment with co-linear structure  is not easy to solve   %Text title% 
%Challenges in Solving Sequence-to-Graph Alignment with Co-Linear Structure
Complexity of Sequence-to-Graph Alignment with Co-Linear Chaining
\footnotetext{\footnotesize
%Received ; Accepted \\ % The received date and accepted date of the article (the author does not need to fill out )%
Supported by Guizhou Provincial Basic Research Program (Natural Science) (Grant No.\,  ZK[2022]020).\\
* Corresponding author\\  % indicate the corresponding author in case of multiple authors, and it is not necessary to do so in case of only one author
 E-mail address: xingfuli@mail.gufe.edu.cn
(Xingfu  Li)
} } \ec  % fund project and author's email address %

\vskip 5mm
\bc{\bf  Xingfu Li$^{*}$}\\  %Author's name (the given name is before the surname),use comma to separate multiple authors and add * after the corresponding author,as Si LI here%
{\small\it  College of Big Data Statistics, Guizhou University of Finance and Economics,  Guiyang $550025$, P. R. China
%\hspace{-2.7cm}$2$. University of Tsukuba, $1$-$1$-$1$ Tennodai, Tsukuba, Ibarake $305$-$8573$, Japan;\\
%\hspace{-0.95cm}$3$. Department of Applied Mathematics, Dalian University, Liaoning $116600$, P. R. China
}\ec   % affiliation of the authors %

\vskip 1 mm

{\narrower\noindent{\small {\small\bf Abstract}\ \   
%Sequence alignment is a fundamental problem used to analyze similarities and differences between biological sequences in computational biology. Sequence-to-graph alignment is one of its variants essential for efficiently addressing genetic variations. In this paper,  we  formulate the  Gap-sensitive Co-Linear Chaining (Gap-CLC) problem and the Co-Linear Chaining with Errors based on Edit Distance (Edit-CLC) problem within the framework of sequence-to-graph alignment, and then study the computational complexity of them.  We demonstrate that solving the Gap-CLC problem in sub-quadratic time is unlikely unless the Strong Exponential Time Hypothesis (SETH) is false. This holds true even for binary alphabets. Additionally, we prove that the Edit-CLC problem is NP-hard when errors are allowed in a graph. These results indicate that models incorporating "co-linear" structures are at least as complex as those without, in the context of sequence-to-graph alignment.
Sequence alignment is a cornerstone technique in computational biology for assessing similarities and differences among biological sequences. A key variant, sequence-to-graph alignment, plays a crucial role in effectively capturing genetic variations. In this work, we introduce two novel formulations within this framework: the Gap-sensitive Co-Linear Chaining (Gap-CLC) problem and the Co-Linear Chaining with Errors based on Edit Distance (Edit-CLC) problem, and we investigate their computational complexity. We show that solving the Gap-CLC problem in sub-quadratic time is highly unlikely unless the Strong Exponential Time Hypothesis  fails—even when restricted to binary alphabets. Furthermore, we establish that the Edit-CLC problem is NP-hard in the presence of errors within the pan-genome graph. These findings emphasize that incorporating co-linear structures into sequence-to-graph alignment models fails to reduce computational complexity, highlighting that these models remain at least as computationally challenging to solve as those lacking such prior information.

%These findings emphasize that incorporating "co-linear" structures into sequence-to-graph alignment models fails to reduce computational complexity, highlighting that these models remain at least as computationally challenging to solve as those lacking such prior information.
% abstract %

\vspace{1mm}\baselineskip 12pt

\no{\small\bf Keywords} \ \ Sequence-to-graph alignment; Co-linear
chaining; Gap-sensitive Co-Linear Chaining; Edit Distance Co-Linear
Chaining; Computational complexity; Strong Exponential Time
Hypothesis; NP-hard % Key words are separated by semicolons%

\par

\vspace{2mm}

\no{\small\bf MR(2020) Subject Classification\ \ {\rm 05C85; 68Q17; 68R10; 92D20; 92-08}} %American Classification Number (more than two numbers are separated by a semicolon without any punctuation at the end)%

}}

\baselineskip 15pt

\section{1. Introduction}

%\subsection{A Subsection}

%Populations undergo genetic mutations and changes in their genomes during the process of evolution, which results in genetic variations among different individuals. By comparing the genetic sequences of different individuals, we can uncover their similarities and differences \cite{genome-similar},  gain insights into the origin, evolution, and adaptability of species \cite{genome-evolution1974On}, and understand the evolutionary relationships and functions of biological sequences \cite{SA-function}. 

Genomic mutations shape the genomes of populations, resulting in genetic divergence among individuals. Consequently, the comparative analysis of genetic sequences enables the identification of genetic similarities and differences \cite{genome-similar}, offers insights into the origin, evolution, and adaptability of species \cite{genome-evolution1974On}, and elucidates the evolutionary relationships and functional dynamics of biological sequences \cite{SA-function}.
Sequence alignment is widely employed in the analysis of DNA, RNA, and protein sequences, serving as a crucial method for exploring gene families, protein structures, and functional characteristics \cite{genome-similar}.
%A lot of work, such as sequence alignment \cite{genome-similar} and Longest Common Subsequence (LCS) \cite{LCS-ESA2024,LCS},   have been done to measure  similarity between two or more sequences. 
%\emph{Sequence alignment} is commonly used  in the analysis of DNA, RNA,  and protein sequences.  It is also an important tool in studying gene families, protein structures, and functions \cite{genome-similar}. 

Sequence alignment aims to identify the optimal correspondence among two or more sequences, revealing their similarity levels and specific variations. By incorporating matches, substitutions, and gaps or insertions, sequences are adjusted to maximize alignment quality according to predefined criteria \cite{SA-operations}.

%The goal of sequence alignment is to find the best matching relationship between two or more sequences and determine the degree of similarity and specific differences between them. During sequence alignment, we can adjust the sequences based on operations such as matches, substitutions, and gaps/insertions, 
%translocations, transpositions, and block-interchanges \cite{transposition}, 
%aiming to achieve optimal similarity under specific alignment rules \cite{SA-operations}.

%To perform sequence alignment, biologists and computer scientists have developed various algorithms and methods. The most classical algorithms include global  and local alignments \cite{Global-local-align}. Global alignment aims to find the best matching relationship across the entire sequences, while local alignment focuses on identifying similarities between segments of the sequences. Sequence alignment often involves computationally intensive algorithms, as it requires processing large amounts of sequence data and searching through all possible alignment combinations. In order to improve alignment efficiency, researchers have developed optimization algorithms and heuristic approaches, such as dynamic programming  \cite{dynamic-1974}, BLAST (Basic Local Alignment Search Tool) \cite{BLAST}, and the Smith-Waterman algorithm \cite{Smith-Waterman-alg}.

Biologists and computer scientists have devised a variety of algorithms to tackle sequence alignment, a fundamental task in bioinformatics. Among the most classical approaches are global and local alignment methods \cite{Global-local-align}. Global alignment seeks to maximize similarity across the entire length of the sequences, providing a comprehensive comparison, whereas local alignment identifies regions of high similarity within subsequences, even if the rest of the sequences differ significantly. 

Given the vast amount of sequence data and the need to evaluate numerous possible alignments, sequence alignment is inherently computationally demanding. To address this challenge, researchers have introduced efficient optimization techniques and heuristic strategies. These include dynamic programming \cite{prousalis2025survey}, which guarantees optimal solutions by breaking the problem into smaller subproblems, the widely used Basic Local Alignment Search Tool \cite{BLAST} that accelerates searches through heuristic filtering, and the sensitive Smith-Waterman algorithm \cite{Smith-Waterman-alg} for precise local alignment. These methods collectively enhance both the accuracy and efficiency of sequence analysis.

%Sequence-to-sequence alignment is naturally derived from  sequence alignment and  has been well studied since the 1970s. Dynamic programming techniques \cite{dynamic-1974} were almost the earliest to solve sequence-to-sequence alignment problems.  Myers and Miller \cite{Multi-alig-SODA} described a multiple-sequence alignment algorithm for determining the highest-scoring alignment, which is based on the line-sweep paradigm, and uses orthogonal range-searching supported by range trees. This result is a time bound higher by a logarithmic factor.  Later in 2005, Abouelhoda and  Ohlebusch \cite{JDA2005} presented better line-sweep algorithms that solve both the global and the local fragment-chaining problem of multiple genomes and reduce both the time-complexity and the space-complexity of Myers and Miller's algorithm.  In recent years gaps and overlaps have been taken into consideration in sequence-to-sequence alignment \cite{Global-local-align,Gap2020,ss-co-linear-overlap-gaps}. 

Sequence-to-sequence alignment, naturally extending from classical sequence alignment, has been extensively studied since the 1970s. Dynamic programming methods \cite{prousalis2025survey} were among the earliest to address this problem. Myers and Miller \cite{Multi-alig-SODA} proposed a multiple-sequence alignment algorithm based on the line-sweep paradigm, leveraging range trees for orthogonal range searching to identify the highest-scoring alignment, though with a time complexity increased by a logarithmic factor. In 2005, Abouelhoda and Ohlebusch \cite{JDA2005} introduced improved line-sweep algorithms that efficiently solve both global and local fragment-chaining problems in multiple genomes, reducing both time and space complexity compared to earlier approaches. Recent years have seen enhanced models incorporating gaps and overlaps in sequence-to-sequence alignment \cite{Global-local-align,Gap2020,ss-co-linear-overlap-gaps}.

%In order to handle massive  genetic variations efficiently, researchers adopted graph model and devised sequence-to-graph alignment problems.  They found that aligning a sequence to a pan-genome is much more practical. A pan-genome is the entire set of genes within a species, consisting of a core genome containing sequences shared among all individuals of the species and  "dispensable" genomes. Pan-genome graphs represent pan-genomes using graph models, effectively capturing the complete genetic variations across the input genomes.  Each vertex in the graphs represents a DNA segment, which can be any length. A path may be a genome, an haplotype, an allele or a variant. Under the help of a pan-genome graph of a species,  the goal of aligning a sequence to a pan-genome graph is to identify a path that best represents the sequence, surpassing other paths in terms of alignment score or genomic coverage.

%To efficiently manage vast genetic sequences and variations, researchers have turned to graph-based models and developed approaches for sequence-to-graph alignment. 
To effectively handle the immense volume of genetic sequences and their variations, researchers have increasingly adopted graph-based models, pioneering innovative methods for aligning sequences  to a pan-genome graph, which has proven particularly practical \cite{cui2025survey}. A pan-genome encompasses all genes within a species, composed of a core genome—shared across all individuals—and variable "dispensable" genomes. Pan-genome graphs, built using graph models, capture the full spectrum of genetic diversity among input genomes. In these graphs, each vertex represents a DNA segment of arbitrary length, while paths correspond to genomes, haplotypes, alleles, or variants. With the aid of a species' pan-genome graph, sequence alignment aims to find the optimal path that best matches the query sequence, evaluated by alignment score or genomic coverage.

%Exact and approximation algorithms for sequence-to-graph alignment problems have been both considered in recent years. If vertices on a path are required to be distinct, the problem is NP-hard \cite{only-path-NPC}  on general pan-genome graphs, even on De Bruijin graphs \cite{single-path-de-Bruijn-complexity},  since it solves the well-known Hamilton Path problem. Most related works do not have this restriction.  Sequence-to-graph alignment on general graphs was studied 20 years ago  as a search problem for hypertext \cite{original}. Amir and Lewenstein  in paper \cite{original} devised an $O(|V| + m|E|)$-time algorithm for exact sequence-to-graph alignment, where  $m$ is the length of a given pattern.  The alignments on DAGs \cite{DAG} and trees \cite{Tree} are both polynomial-time solvable with time complexity $O(|V| + m|E|)$ and $O(|V|)$ respectively. Equi et al. \cite{Eaxct-lower-bound} have recently proved a lower bound saying that, for any constant $\gamma > 0$, unless the Strong Exponential Time Hypothesis (SETH) is false, there is neither $O(|E|^{1 - \gamma}m)$ nor $O(|E|m^{1 - \gamma})$ time algorithm for the exact case on general graphs, even with a binary alphabet. 

In recent years, both exact and approximation algorithms for sequence-to-graph alignment have received significant attention. When the alignment path is constrained to visit distinct vertices, the problem becomes NP-hard \cite{only-path-NPC} on general pangenome graphs, including De Bruijn graphs \cite{single-path-de-Bruijn-complexity}, as it encompasses the classical Hamiltonian Path problem. However, most existing approaches relax this constraint. Sequence-to-graph alignment on general graphs was initially explored two decades ago in the context of hypertext search \cite{original}, where Amir and Lewenstein proposed an $O(|V| + m|E|)$-time algorithm for exact alignment, with $m$ denoting the pattern length. For restricted graph classes, efficient solutions exist: polynomial-time algorithms apply to Directed Acyclic Graphs  (DAGs) \cite{DAG} and trees \cite{Tree}, with time complexities of $O(|V| + m|E|)$ and $O(|V|)$, respectively. Recently, Equi et al. \cite{Eaxct-lower-bound} established a conditional lower bound, demonstrating that, for any constant $\gamma > 0$, unless the Strong Exponential Time Hypothesis (SETH) fails, no algorithm can solve the exact alignment problem in $O(|E|^{1 - \gamma}m)$ or $O(|E|m^{1 - \gamma})$ time on general graphs—even over a binary alphabet.

%The approximation case is a little  complex than the exact case. If errors  are  allowed  on a query sequence rather than on a graph, then this problem is polynomial-time solvable not only on DAGs \cite{DAG-approx} but also on general graphs \cite{original}. Otherwise, it is NP-hard, even the alphabet is binary \cite{original,binary-alphabet}. On De-Bruijin graph, Gibney et al. proved that there is no sub-quadratic time algorithm in the case that substitutions are only allowed on the pattern \cite{approx-hard-de-Bruijn}. 

The approximation scenario is somewhat more intricate than the exact one. When errors are permitted on the query sequence rather than on the graph structure, the problem becomes polynomial-time solvable—not only on DAGs \cite{DAG-approx} but also on general graphs \cite{original}. In contrast, it remains NP-hard even over a binary alphabet \cite{original,binary-alphabet}. For De Bruijn graphs, Gibney et al. demonstrated that no sub-quadratic time algorithm exists when substitutions are restricted to the pattern alone \cite{approx-hard-de-Bruijn}.

%Due to the obstacle of sub-quadratic term in the time-complexity, processing high-throughput sequencing data in a graph model is challenging in practice. Recently seed-chain-extended has been commonly used by modern alignment tools \cite{long-read-mapping2022,Minimap2018,long-read-survey-2022}. Co-linear chaining is a problem in the process of seed-chain-extended. 

The sub-quadratic term in time complexity poses a significant obstacle, making it practically challenging to process high-throughput sequencing data within a graph-based framework. Recently, the seed-chain-extension approach has become widely adopted in modern alignment tools \cite{long-read-mapping2022,Minimap2018,long-read-survey-2022}, where co-linear chaining emerges as a key computational challenge.

%Co-linear chaining is not  new  and  has been well-studied in sequence-to-sequence alignment \cite{Multi-alig-SODA,sum-of-gap,ss-co-linear-overlap-gaps}. Besides a query string and a reference sequence, an instance of a co-linear chaining problem additionally contains  a set of anchors. An anchor can be viewed as a pair of intervals in  two sequences corresponding to the exact seed match. An anchor-chain is an ordered subset of anchors whose intervals must appear in increasing order in both sequences. The co-linear chaining problem seeks an anchor-chain with the highest score, where the score of a chain is calculated by summing the weights of the anchors in the chain and subtracting the penalty for gaps between adjacent anchors. The score functions  are established usually based on gap-costs \cite{sum-of-gap} and overlaps \cite{ss-co-linear-overlap-gaps}.

Co-linear chaining is a well-established concept, extensively studied in sequence-to-sequence alignment \cite{Multi-alig-SODA,sum-of-gap,ss-co-linear-overlap-gaps}. Beyond the query string and reference sequence, a co-linear chaining problem includes a set of anchors—each representing a pair of intervals in the two sequences that correspond to exact seed matches. A valid anchor-chain is an ordered subset of anchors, requiring their intervals to appear in strictly increasing order in both sequences. The goal is to find the highest-scoring chain, where the score combines anchor weights and penalizes gaps between consecutive anchors. Scoring schemes typically incorporate gap costs \cite{sum-of-gap} and may account for overlaps \cite{ss-co-linear-overlap-gaps}.

%However, co-linear chaining in sequence-to-graph alignment is not easier  than that in sequence-to-sequence alignment. M\"{a}kinen et al. \cite{first-to-graph} studied co-linear chaining problems on directed acyclic graphs (DAGs), and devised a sparse dynamic programming algorithm on  DAGs. Later gap cost was introduced into the score function \cite{co-linear-gap}. Recently Rajput et al. \cite{co-linear-general} have presented a practical formula and an algorithm for co-linear chaining problem on general graphs. The formula considers gap costs not only on  query strings, but also on  pan-genome graphs.  

%However, co-linear chaining in sequence-to-graph alignment remains not so simpler than that in sequence-to-sequence alignment. 
Nevertheless, co-linear chaining in sequence-to-graph alignment remains more complex than in sequence-to-sequence alignment, posing greater computational challenges.
M\"{a}kinen et al. \cite{first-to-graph} investigated co-linear chaining on DAGs, developing a sparse dynamic programming approach tailored to DAGs. Subsequently, gap costs were incorporated into the scoring function \cite{co-linear-gap}. More recently, Rajput et al. \cite{co-linear-general} proposed a practical formulation and an efficient algorithm for co-linear chaining on general graphs, extending gap cost considerations to both query sequences and pan-genome graphs.

%In this paper, we pay attention to the computational complexity of Gap-sensitive Co-Linear Chaining problem (Gap-CLC) and Co-Linear Chaining with errors based on Edit distance problem (Edit-CLC) in sequence-to-graph alignment.  An anchor and an anchor-chain  in sequence-to-graph alignment are  a little  complex than those in sequence-to-sequence alignment.  We define an anchor by use of  Cartesian Product in Section \ref{Notations}. Specially, rather than an explicit tuple, an anchor is defined as an element in a Cartesian Product from appearances of a sub-sequence respectively on the given query sequence and pan-genome graph. Due to this, an instance of the Gap-CLC problem, as well as the Edit-CLC problem, contains a set of sub-sequences and two other sets respectively indicating appearances of a sub-sequence on the given query sequence and pan-genome graph. Both of them are formulated respectively as Problems \ref{Gap-CLC-form} and \ref{Edit-CLC-form}. Section \ref{complexity} presents the hardness of Gap-CLC problem and Edit-CLC problem. Intuitively speaking, the results show that  models that do not incorporate "co-linear" structures are respectively reduced to gap-sensitive co-linear chaining  and co-linear chaining with errors, which indicates that models incorporating "co-linear" structures are at least as complex as models that do not. Finally we conclude this paper and present some future work in Section \ref{conclude}.

In this paper, we investigate the computational complexity of the Gap-sensitive Co-Linear Chaining (Gap-CLC) and Co-Linear Chaining with Errors based on Edit distance (Edit-CLC) in sequence-to-graph alignment. Anchors and anchor-chains in this context are more intricate than their sequence-to-sequence counterparts. We formalize an anchor using the Cartesian Product, as defined in Section \ref{Notations}, where an anchor is represented not as an explicit tuple, but as an element of a Cartesian Product derived from the occurrences of a sub-sequence on both the query sequence and the pan-genome graph. Consequently, an instance of either the Gap-CLC or Edit-CLC problem comprises a set of sub-sequences along with two associated sets capturing their occurrences on the query and graph. These are formally stated as Problems \ref{Gap-CLC-form} and \ref{Edit-CLC-form}, respectively. Section \ref{complexity} analyzes the computational complexity of both problems. Intuitively, models ignoring "co-linear" structures can be reduced to the Gap-CLC and Edit-CLC frameworks, implying that incorporating co-linearity does not reduce complexity and may in fact preserve or increase it. We conclude with a summary in Section \ref{conclude} and propose some future work in Section \ref{discussions}.

\section{2. Basic Terminologies and Problem Definitions}
\label{Notations}
%This section initially defines an anchor and an anchor-chain using a Cartesian Product. Subsequently, we define the sequence spelled by an anchor-chain in sequence-to-graph alignment. Finally, we formulate the Gap-CLC  and Edit-CLC problems. 
This section begins by formally defining  \emph{anchor} and  \emph{anchor-chain} through the Cartesian Product. It then proceeds to characterize the sequence encoded by an anchor-chain within the context of sequence-to-graph alignment. Finally, it presents the formal formulations of the Gap-CLC and Edit-CLC problems.

%A \emph{pan-genome graph} $G=(V,E,\delta)$ is a directed and vertex-labeled graph in which $V$ is the vertex set; $E$ is the directed edge set;  $\delta: V\rightarrow \Sigma^*\setminus \{\epsilon\}$ is a function from $V$ to the set of all finite strings $\Sigma^{*}$ on a given alphabet $\Sigma$, where $\epsilon$ is an empty string.  In this paper, the set of vertices and edges are respectively denoted as $V(G)$ and $E(G)$. In addition, the string corresponding to a vertex $v\in V(G)$ is written as $\delta(v)$.   Since a read might appear continuously over multiple locations, this paper allows \emph{self-loops} in a pan-genome graph. In other words, there might be an edge $(v,v)\in E(G)$ for a vertex $v$ in a pan-genome graph $G$. 

A \emph{pan-genome graph} $G = (V, E, \delta)$ is a directed graph with vertex labels, where $V$ denotes the set of vertices, $E$ denotes the set of directed edges, and $\delta: V \to \Sigma^* \setminus \{\epsilon\}$ assigns to each vertex a non-empty string from the finite alphabet $\Sigma$, with $\epsilon$ representing the empty string. We denote the vertex and edge sets of $G$ as $V(G)$ and $E(G)$, respectively. For any vertex $v \in V(G)$, the associated string is given by $\delta(v)$. As reads may span multiple consecutive positions in the graph, we permit \emph{self-loops}, i.e., an edge $(v, v) \in E(G)$ may exist for some vertex $v$ in $G$.

%The number of elements in a sequence $S$  is said to be the \emph{length} of it and denoted as $|S|$. The subscripts are the indices and indicate the positions of the elements. In this paper, all indices begin with zero. For example, in a sequence  $S=$"$ababc$", the leftmost '$a$' is indexed as $S_0$, while the rightmost '$a$' is indexed as $S_2$. If $a$ is an element of a sequence $S$, then we denote this as "$a\in S$". If there is a walk from $u$ to $v$ in $G$, then we say that "$u$ reaches $v$" or "$v$ can be reached from $u$".   

The number of elements in a sequence $S$ is referred to as its \emph{length}, denoted by $|S|$. Indices, represented as subscripts, indicate the positions of elements within the sequence, and in this paper, all indexing starts at zero. For instance, in the sequence $S = "ababc"$, the leftmost '$a$' is indexed as $S_0$, and the rightmost '$a$' is indexed as $S_2$. 
%We write $a \in S$ to denote that $a$ is an element of $S$. 
Furthermore, if there exists a walk from vertex $u$ to vertex $v$ in a graph $G$, we say that "$u$ reaches $v$" or equivalently that "$v$ is reachable from $u$".

%Given two sequences  $q\in \Sigma^{*}\setminus\{\epsilon\}$ and  $s\in \Sigma^{*}\setminus\{\epsilon\}$, if there is an integer $x: \ 0\leq x< |s|$ such that $q_{i} = s_{x+i}$ and $x + i <|s|$ for every $0\leq i < |q|$, then the integer $x$ is said to be a \emph{location} of $q$ in $s$. And $q$ is said to be \emph{located} at $x$ in $s$.  Obviously, there might be more than one location  of   $q$ in $s$. For example, if $q=$"$ab$" and $s=$"$ababc$", then $q$ is located at $0$ as well as at $2$ in $s$.  Let $loc(q,s)$ be the set of locations of $q$ in $s$. 

Given two non-empty sequences $q \in \Sigma^{*}\setminus\{\epsilon\}$ and $s \in \Sigma^{*}\setminus\{\epsilon\}$, an integer $x$ with $0 \leq x < |s|$ is called a \emph{location} of $q$ in $s$ if $q_i = s_{x+i}$ for all $0 \leq i < |q|$, provided $x+i < |s|$. In this case, we say $q$ is \emph{located} at position $x$ in $s$. Clearly, $q$ may occur at multiple locations within $s$. For instance, when $q = ab$ and $s = ababc$, $q$ appears at positions $0$ and $2$. 
%Denote by $loc(q,s)$ the set of all such locations of $q$ in $s$.
Let $\mathrm{loc}(q,s)$ denote the collection of all positions where $q$ appears in $s$.

%Given a sequence $q\in \Sigma^{*}\setminus\{\epsilon\}$ and a pan-genome graph $G$, if there is a vertex $v\in V(G)$ and an integer $x:\ 0\leq x<|\delta(v)|$ such that $q$ is located at $x$ in $\delta(v)$, then the ordered pair $(v,x)$ is said to be a location of $q$ in $G$. And $q$ is said to be located at $(v,x)$ in $G$. Let $loc(q,G)$ be the set of locations of $q$ in a pan-genome graph $G$. For each location $p\in loc(q,G)$,  the first element of $p$ is denoted as $p\cdot v$, while the second one  is denoted as $p\cdot x$. 

Given a non-empty sequence $q \in \Sigma^{*} \setminus \{\epsilon\}$ and a pan-genome graph $G$, 
%a vertex $v \in V(G)$ together with an integer $x:$ $0 \leq x < |\delta(v)|$ defines a position where $q$ located at  $x$ within $\delta(v)$. In such a case, the ordered pair $(v, x)$ is referred to as a location of $q$ in $G$, 
 if there is a vertex $v\in V(G)$ and an integer $x:\ 0\leq x<|\delta(v)|$ such that $q$ is located at $x$ in $\delta(v)$, then the ordered pair $(v,x)$ is said to be a location of $q$ in $G$.
And we say $q$ is located at $(v, x)$ in $G$. Let $loc(q, G)$  be the set of all  locations of $q$ in the pan-genome graph $G$. For every location $p \in loc(q, G)$, the vertex component is denoted $p \cdot v$ and the offset component $p \cdot x$.

%Given two sequences $q, Q\in \Sigma^{*}\setminus\{\epsilon\}$ and a pan-genome graph $G$, let $A_q$ be the \emph{cartesian product} of $loc(q,Q)$ and $loc(q,G)$, i.e.,  $A_q = loc(q,Q)\times loc(q,G) = \{(x,y):\ x\in loc(q,Q)$ and  $y\in loc(q,G)\}$. Note that if either $loc(q,Q)=\emptyset$ or $loc(q,G)=\emptyset$, then $A_q=\emptyset$. Each element belonging to $\cup_{q\in \Sigma^{*}\setminus \{\epsilon\}}A_q$ is said to be an \emph{anchor}. Every anchor corresponds to a unique sequence $q\in \Sigma^{*}\setminus\{\epsilon\}$. For an anchor $a=(x,y)$ in the set  $\cup_{q\in \Sigma^{*}\setminus \{\epsilon\}}A_q$ , the first element of $a$ is denoted as $F(a)$, while the second one is denoted as $H(a)$, i.e., $ F(a)= x$ and $ H(a)=y$. A subset of $\cup_{q\in \Sigma^{*}\setminus \{\epsilon\}}A_q$ is named as an \emph{anchor-set} on $(Q,G)$.

Given two non-empty sequences $q, Q \in \Sigma^{*}\setminus\{\epsilon\}$ and a pan-genome graph $G$, define $A_q$ as the \emph{cartesian product} of $loc(q,Q)$ and $loc(q,G)$, that is, $A_q = loc(q,Q) \times loc(q,G) = \{(x,y) \mid x \in loc(q,Q),\, y \in loc(q,G)\}$. It follows that $A_q = \emptyset$ if either $loc(q,Q) = \emptyset$ or $loc(q,G) = \emptyset$. Each element in $\bigcup_{q\in \Sigma^{*}\setminus \{\epsilon\}} A_q$ is called an \emph{anchor}, where every anchor is uniquely associated with a sequence $q \in \Sigma^{*}\setminus\{\epsilon\}$. For any anchor $a = (x,y)$, denote its first component as $F(a) = x$ and its second as $H(a) = y$. A subset of $\bigcup_{q\in \Sigma^{*}\setminus \{\epsilon\}} A_q$ is referred to as an \emph{anchor-set} on $(Q,G)$.

%Let $Q\in \Sigma^{*}\setminus\{\epsilon\}$ be a sequence which is going to be aligned with a pan-genome graph $G$. Let $A$ be an  anchor-set on $(Q,G)$. There is  a pre-order relation $\precsim_{q}$  among anchors in $A$ with respect to the query sequence $Q$.  An anchor $a_{i}\in A$ \emph{precedes} another one $a_{j}\in A$ with respect to $Q$, i.e., $a_{i}\precsim_{q} a_{j}$, if and only if $F(a_i) < F(a_j)$. With respect to  $G$,  another pre-order relation $\precsim_{g}$ exists among anchors in $A$. An anchor $a_{i}\in A$ \emph{precedes} another one $a_{j}\in A$ with respect to the  graph $G$, i.e., $a_{i}\precsim_{g}a_{j}$,  if and only if either of the following two propositions holds: (1)  $H(a_i)\cdot v\not = H(a_j)\cdot v$ but $H(a_{i})\cdot v$ can reach  $H(a_{j})\cdot v$ in $G$; (2) $H(a_i)\cdot v = H(a_j)\cdot v$ and $H(a_i)\cdot x <  H(a_j)\cdot x$.

Let $Q \in \Sigma^{*} \setminus \{\epsilon\}$ be a non-empty sequence to be aligned with a pan-genome graph $G$, and let $A$ be an anchor-set defined on $(Q, G)$. A pre-order relation $\precsim_{q}$ is established among anchors in $A$ with respect to $Q$: an anchor $a_i \in A$ precedes $a_j \in A$, denoted $a_i \precsim_{q} a_j$, if and only if $F(a_i) < F(a_j)$. Similarly, with respect to $G$, a pre-order relation $\precsim_{g}$ exists on $A$. We say $a_i \precsim_{g} a_j$ if either: (1) $H(a_i)\cdot v \neq H(a_j)\cdot v$ and there exists a path from $H(a_i)\cdot v$ to $H(a_j)\cdot v$ in $G$, or (2) $H(a_i)\cdot v = H(a_j)\cdot v$ and $H(a_i)\cdot x < H(a_j)\cdot x$. These relations induce a hierarchical ordering of anchors reflecting their sequential and graphical consistency.

%For two anchors $a_i$ and $a_j$, if   $a_i\precsim_{q}a_j$ and $a_i\precsim_{g}a_j$ both hold, then this is denoted as $a_{i}\precsim a_{j}$. A sequence of anchors S is an \emph{anchor-chain} if and only if  $S_i\precsim S_{i+1}$ for every $0\leq i < |S|-1$.

For two anchors $a_i$ and $a_j$, we write $a_i \precsim a_j$ if both $a_i \precsim_{q} a_j$ and $a_i \precsim_{g} a_j$ hold. A sequence of anchors $S$ is called an \emph{anchor-chain} if and only if $S_i \precsim S_{i+1}$ for all $0 \leq i < |S|$.

%Let $Q\in\Sigma^{*}\setminus\{\epsilon\}$ and $G$ respectively be a sequence and a pan-genome graph. Let $W$ be a walk in $G$.  The sequence $\sum_{i=0}^{|W|-1}\delta(W_i) = \delta(W_0) + \delta(W_1) + \dots + \delta(W_{|W|-1})$ is said to be \emph{spelled by $W$} in the pan-genome graph, where "$+$" is the concatenation of two strings.  Every walk in a pan-genome graph naturally spells a sequence. 

Let $Q \in \Sigma^{*} \setminus \{\epsilon\}$ and $G$ denote a sequence and a pan-genome graph, respectively, and let $W$ be a walk in $G$. The sequence defined by $\sum_{i=0}^{|W|-1} \delta(W_i) = \delta(W_0) + \delta(W_1) + \cdots + \delta(W_{|W|-1})$ is said to be \emph{spelled by $W$} in the pan-genome graph, where the operation "$+$" represents string concatenation. Notably, every walk in a pan-genome graph inherently corresponds to a spelled sequence.

%Recall that the \emph{Sequence-to-Graph Exact Matching }(Exa-SGM) Problem is to find a walk in a graph so that the sequence spelled by this walk is the same as a given  query sequence.  The \emph{Sequence-to-Graph Matching with errors }(Err-SGM) problem generalizes the Exa-SGM problem, which asks to find a walk spelling a sequence with minimum edit distance to a query sequence.   If each vertex is labeled with a single character, then the Exa-SGM and Err-SGM problems are respectively called \emph{Single-Exa-SGM} and \emph{Single-Err-SGM}.

Recall the \emph{Sequence-to-Graph Exact Matching} (Exa-SGM) Problem: it seeks a walk in a graph such that the sequence spelled by the vertices along this walk exactly matches a given query sequence. In contrast, the \emph{Sequence-to-Graph Matching with Errors} (Err-SGM) Problem extends Exa-SGM by aiming to find a walk whose spelled sequence minimizes the edit distance to the query sequence. When each vertex is labeled with a single character, these variants are termed \emph{Single-Exa-SGM} and \emph{Single-Err-SGM}, respectively.

\subsection{Problem Formulations}
%Besides a query sequence $Q\in \Sigma^{*}\setminus\{\epsilon\}$, a pan-genome graph $G = (V, E, \delta)$ in an instance of the sequence-to-graph alignment problem, \emph{Co-Linear Chaining (CLC) problem in sequence-to-graph alignment} needs an additional anchor-set. In previous work, such as \cite{Multi-alig-SODA,sum-of-gap,ss-co-linear-overlap-gaps}, an anchor-set is explicitly given by a set of  tuples. However, in this work, an anchor-set is implicitly defined by  Cartesian Product from appearances of some sequences respectively in $Q$ and $G$.

Besides a query sequence $Q \in \Sigma^{*} \setminus \{\epsilon\}$, an instance of the sequence-to-graph alignment problem—specifically, the \emph{Co-Linear Chaining (CLC) problem in sequence-to-graph alignment} requires a pan-genome graph $G = (V, E, \delta)$ and an additional anchor-set. Prior studies, such as \cite{Multi-alig-SODA,sum-of-gap,ss-co-linear-overlap-gaps}, explicitly define the anchor-set as a collection of tuples. In contrast, this work introduces an implicit definition of the anchor-set through the Cartesian Product of occurrences of certain sequences in $Q$ and $G$, respectively.

%An anchor-set $A$ can be partitioned into several subsets so that in each subset the sequence corresponding to every anchor is the same. Without loss of generality,  the number of subsets is denoted as $k$. Let $A = \bigcup_{i=1}^{k}A_i$ in which  $A_i$ is a subset of $A$ such that for $1\leq i\leq k$ the sequence corresponding to every anchor in $A_i$ is $q_i$, and $q_i\not = q_j$ for every two distinct $A_i$ and $A_j$. Let $R=\bigcup_{i=1}^{k}\{q_i\}$. 

An anchor-set $A$ can be decomposed into $k$ disjoint subsets, where each subset contains anchors associated with an identical sequence. Specifically, we write $A = \bigcup_{i=1}^{k} A_i$, such that all anchors in $A_i$ correspond to the same sequence $q_i$, and $q_i \ne q_j$ whenever $i \ne j$. Thus, the set $R = \bigcup_{i=1}^{k} \{q_i\}$ collects all distinct sequences induced by the partition of $A$.

\begin{lemma}\label{implicit-is-possible}
  
        Given an alphabet $\Sigma$, a query sequence $Q\in \Sigma^{*}\setminus\{\epsilon\}$, a pan-genome graph $G = (V, E, \delta)$ and a set of anchors $A$ on $(Q,G)$,  there are two additional sets $X_q\subseteq loc(q,Q)$ and $Y_q\subseteq loc(q,G)$ for each sequence $q\in R$, such that $A \subseteq \bigcup_{q\in R}X_q\times Y_q$.
    
\end{lemma}
\begin{proof}
        Suppose that $A$ can be partitioned into $k$ subsets, i.e.,  $A = \bigcup_{i=1}^{k}A_i$, such that  the sequence corresponding to every anchor in  $A_i$ is $q_i$ for $1\leq i\leq k$.  According to the definitions of anchors, we have $A_i\subseteq A_{q_i} = loc(q_i,Q)\times loc(q_i,G)$ for $1\leq i\leq k$. Then there are two sets $X_{q_i}\subseteq loc(q_i,Q)$ and $Y_{q_i}\subseteq loc(q_i,G)$ such that $A_i \subseteq X_{q_i}\times Y_{q_i}$ for $1\leq i\leq k$. Therefore, we have $A \subseteq \bigcup_{i=1}^{k} X_{q_i}\times Y_{q_i} = \bigcup_{q\in R}X_q\times Y_q$.
\end{proof}

    By Lemma \ref{implicit-is-possible}, each explicit anchor-set is  a subset of $\bigcup_{q\in R}X_q\times Y_q$, where $R$ is the set of sequences in the given anchor-set, $X_q\subseteq loc(q,Q)$ and $Y_q\subseteq loc(q,G)$.

%The Gap-CLC problem in this work  is given  a query sequence $Q\in \Sigma^{*}\setminus\{\epsilon\}$, a pan-genome graph $G = (V, E, \delta)$,  a set of sequences $R\subseteq \Sigma^{*}$,   two additionally sets: $X_q\subseteq loc(q,Q)$ and $Y_q\subseteq loc(q,G)$ for  each sequence $q\in R$, and  a binary gap-function $f(Gap_{Q}(S), Gap_{G}(S))$.  The aim is to find an anchor-chain $S$ from $\bigcup_{q\in R}X_q\times Y_q$ to maximize $|S|$ $-$ $f(Gap_{Q}(S), Gap_{G}(S))$     where  $Gap_Q(S)$ and $Gap_G(S)$ are two unary gap-cost functions respectively with respect to $Q$ and  $G$. Please see Problem \ref{Gap-CLC-form} for the formulation of Gap-CLC problem. If each vertex is labeled with a single character, then it is named as \emph{Single-Gap-CLC} problem.

The Gap-CLC problem addressed in this work involves a query sequence $Q \in \Sigma^{*} \setminus \{\epsilon\}$, a pan-genome graph $G = (V, E, \delta)$, and a set of sequences $R \subseteq \Sigma^{*}$. For each sequence $q \in R$, we are given two additional sets: $X_q \subseteq loc(q, Q)$ and $Y_q \subseteq loc(q, G)$, along with a binary gap-function $f(Gap_{Q}(S), Gap_{G}(S))$. The objective is to find an anchor-chain $S$ from $\bigcup_{q \in R} X_q \times Y_q$ that maximizes $|S| - f(Gap_{Q}(S), Gap_{G}(S))$, where $Gap_Q(S)$ and $Gap_G(S)$ denote unary gap-cost functions with respect to $Q$ and $G$, respectively. See Problem \ref{Gap-CLC-form} for the formal formulation of the Gap-CLC problem. When each vertex in the graph is labeled with a single character, the problem is referred to as the \emph{Single-Gap-CLC} problem.

\begin{table}[!ht]
    \renewcommand{\arraystretch}{1.2}
    \centering
    \begin{tabular}{l l }
        \toprule
        \textbf{Instance}: & An alphabet $\Sigma$; a query sequence $Q$ over $\Sigma$; a pan-genome graph \\
                            & $G=(V,E,\delta)$   where $\delta: V\rightarrow\Sigma$; a set of sequences $R\subseteq \Sigma^{*}$; for \\
                            & each sequence $q\in R$ there are  two additionally sets: $X_q\subseteq loc(q,Q)$ \\
                            & and $Y_q\subseteq loc(q,G)$; a binary gap-function $f(Gap_{Q}(S), Gap_{G}(S))$. \\
        \midrule
        \textbf{Query}: & Find an anchor-chain $S$ in $\bigcup_{q\in R}X_q\times Y_q$   such that \\
                        & $|S|-f(Gap_{Q}(S), Gap_{G}(S))$ is the maximum. \\
        \bottomrule
    \end{tabular}
     \caption{\textbf{Gap-CLC problem}}\label{Gap-CLC-form}
\end{table}

%The Edit-CLC problem is given  a query sequence $Q\in \Sigma^{*}\setminus\{\epsilon\}$, a pan-genome graph $G = (V, E, \delta)$, a set of sequences $R\subseteq \Sigma^{*}$,   two additionally sets: $X_q\subseteq loc(q,Q)$ and $Y_q\subseteq loc(q,G)$ for  each sequence $q\in R$.  The problem  is asked to find an anchor-chain $S$ from $\bigcup_{q\in R}X_q\times Y_q$   so that the edit distance between $Q$ and the sequence spelled by $S$ is minimum, where one can refer to Section \ref{walk-chain} for the sequence spelled by an anchor-chain. Please see Problem \ref{Edit-CLC-form} for the formulation of Edit-CLC problem. If each vertex is labeled with a single character, then the special case of Edit-CLC problem is named as \emph{Single-Edit-CLC} problem.

The Edit-CLC problem is defined as follows: given a query sequence $Q \in \Sigma^{*} \setminus \{\epsilon\}$, a pan-genome graph $G = (V, E, \delta)$, and a set of sequences $R \subseteq \Sigma^{*}$, along with two associated sets $X_q \subseteq loc(q,Q)$ and $Y_q \subseteq loc(q,G)$ for each $q \in R$, the goal is to find an anchor-chain $S$ drawn from $\bigcup_{q\in R} X_q \times Y_q$ such that the edit distance between $Q$ and the sequence induced by $S$ is minimized. 
%The sequence spelled by an anchor-chain is formally explained in Section~\ref{walk-chain}; 
Problem~\ref{Edit-CLC-form} proposes a precise formulation of the Edit-CLC problem. When every vertex in the graph is labeled with a single character, this special instance of the Edit-CLC problem is referred to as the \emph{Single-Edit-CLC} problem.

\begin{table}[!ht]
    \renewcommand{\arraystretch}{1.2}
    \centering
    \begin{tabular}{l l }
        \toprule
        \textbf{Instance}: & An alphabet $\Sigma$; a query sequence $Q$ over $\Sigma$; a pan-genome graph \\
                            & $G=(V,E,\delta)$   where $\delta: V\rightarrow\Sigma$; a set of sequences $R\subseteq \Sigma^{*}$; for \\
                            & each sequence $q\in R$ there are  two additionally sets: $X_q\subseteq loc(q,Q)$ \\
                            & and $Y_q\subseteq loc(q,G)$.  \\
        \midrule
        \textbf{Query}: & Find an anchor-chain $S$ in $\bigcup_{q\in R}X_q\times Y_q$   such that the \\
                        & distance between $Q$ and  the sequence spelled by $S$ is minimum. \\
        \bottomrule
    \end{tabular}
     \caption{\textbf{Edit-CLC problem}}\label{Edit-CLC-form}
\end{table}

\section{3. The Complexity}
\label{complexity}
%In this section, every vertex is labeled with a single character in a graph,because it suffices to use single symbol to show that Gap-CLC and Edit-CLC in graphs are challenging. We establish a linear-time reduction from the Single-Exa-SGM problem to the Gap-CLC problem and prove that there is  no sub-quadratic time algorithm to solve the Cap-CLC problem, since the Single-Exa-SGM problem does not admit  sub-quadratic time algorithm, even with a binary alphabet \cite{Eaxct-lower-bound}. Concretely, the lower-bound for Single-Exa-SGM problem is as following. For any constant $\gamma > 0$, the Single-Exa-SGM problem can not be solved in  $O(|E|^{1-\gamma}m) + O(|E|m^{1-\gamma})$ time unless the Strong Exponential Time Hypothesis (SETH) is false, even on a binary alphabet \cite{Eaxct-lower-bound}, where $E$ and $m$ are respectively the edge set of a  graph and the length of a query sequence.  

In this section, each vertex in the graph is assigned a single-character label, which is sufficient to demonstrate the inherent difficulty of the Gap-CLC and Edit-CLC problems. We present a linear-time reduction from the Single-Exa-SGM problem to Gap-CLC, establishing that no sub-quadratic time algorithm exists for Gap-CLC, unless the Single-Exa-SGM problem admits such an algorithm, which is ruled out even over a binary alphabet \cite{Eaxct-lower-bound}. Specifically, for any constant $\gamma > 0$, the Single-Exa-SGM problem cannot be solved in either $O(|E|^{1-\gamma}m)$ or $O(|E|m^{1-\gamma})$ time unless the SETH fails, where $E$ denotes the edge set of the graph and $m$ is the length of the query sequence \cite{Eaxct-lower-bound}.

%The input of Single-Exa-SGM problem is a query string $Q\in \Sigma^{*}\setminus\{\epsilon\}$ and a vertex-labeled graph $G$, where each vertex of $G$ is labeled by a single character from the alphabet $\Sigma = \{0,1\}$.  Let $|Q|$ be the length of $Q$ and $|V(G)|$  be the number of vertices in $G$. The  gap-function $f(Gap_{Q},Gap_{G})$ is defined as following. 

The Single-Exa-SGM problem takes as input a non-empty query string $Q \in \Sigma^{*} \setminus \{\epsilon\}$ and a vertex-labeled graph $G$, with each vertex assigned a label from the binary alphabet $\Sigma = \{0,1\}$. Let $|Q|$ denote the length of $Q$ and $|V(G)|$ represent the number of vertices in $G$. The gap-function $f(Gap_{Q}, Gap_{G})$ is defined as follows.

\begin{enumerate}
    %$ weight(a_i) = N_q + N_g$;
    \item %The gap function $f(Gap_{Q}(S), Gap_{G}(S))$ is $Gap_{Q}(S) + Gap_{G}(S)$;
    The gap function $f(Gap_{Q}(S), Gap_{G}(S))$ is defined as the sum $Gap_{Q}(S) + Gap_{G}(S)$, capturing the combined discrepancy between the two gap measures;
    \item The two gap measures are respectively defined as: $Gap_{Q}(S) = \Sigma_{0\leq i < |S|-1}gap_{Q}(S_i,S_{i+1})$ and $Gap_{G}(S) = \Sigma_{0\leq i < |S|-1}gap_{G}(S_i,S_{i+1})$;
    \item %The gap cost $gap_{Q}(S_i, S_{i+1}) = F(S_{i+1})-F(S_i)-1$ for every $0\leq i < |S|-1$;
    The gap cost associated with $Q$, defined as $\text{gap}_Q(S_i, S_{i+1}) = F(S_{i+1}) - F(S_i) - 1$, is computed for each consecutive pair $(S_i, S_{i+1})$ with $0 \leq i < |S| - 1$, capturing the incremental difference in the function $F$ minus a unit offset;
    %is defined as  the number of characters between  $s_{i}^{Q}$ and  $s_{i+1}^{Q}$ in $Q$ (excluding $s_{i}^{Q}$ and  $s_{i+1}^{Q}$).
    \item %For every $0\leq i < |S|-1$, the gap cost $gap_{G}(S_i, S_{i+1}) = 0$, if $H(S_{i+1})\cdot v$ is a neighbor of $H(S_i)\cdot v$. Otherwise, $gap_{G}(S_i, S_{i+1}) = |Q| + |V(G)|$, 
    For each $0 \leq i < |S| - 1$, the gap cost associated with $G$ is defined as $gap_{G}(S_i, S_{i+1}) = 0$ if $H(S_{i+1}) \cdot v$ is adjacent to $H(S_i) \cdot v$ in the graph; otherwise, $gap_{G}(S_i, S_{i+1}) = |Q| + |V(G)|$, 
    %\item The parameter $k = N_q$;
\end{enumerate}
where $S$ is an anchor-chain.

%The reduction from Single-Exa-SGM to Gap-CLC is as following.
The reduction from Single-Exa-SGM to Gap-CLC proceeds as follows.
\begin{enumerate}
    \item %The alphabet remains to be $\{0,1\}$. In addition, the query sequence and the pan-genome graph also remain to be $Q$ and $G$; 
    The alphabet continues to be $\{0,1\}$, and likewise, the query sequence and the pan-genome graph remain denoted as $Q$ and $G$, respectively;
    \item {The set of sequences is  $R = \{0,1\}$};
    \item {Let $X_{q} = loc(q,Q)$ and $Y_q = loc(q,G)$ for every $q\in R$};
    %\item {The weight of every anchor $a \in X_q\times Y_q$ is $ weight(a) = 1$, where $q\in R$.}
\end{enumerate}

{It is easy to find that the reduction takes linear-time.}
Recall that the sequence adhering to every anchor  is a single character. 
%By Lemma \ref{complete-q-cluster} and Definition \ref{anchor-set-def}, 
%in order to earn the anchor-set, it is sufficient for us to store the locations for each of $'A'$, $'C'$, $'G'$ and $'T'$ respectively on $Q$ and $G$. 
%Additionally the  anchor-set can be represented by the locations of $0$ and $1$ respectively on $Q$ and $G$. 
% by traversing $Q$ and $G$. 
So both $loc(q,Q)$ and $loc(q,G)$ for $q\in R$ can be obtained by traversing $Q$ and $G$ in linear time. 
%by use of Algorithm \ref{anchor-set}.

\iffalse
\begin{algorithm}[h!]
%\floatname{algorithm}{Algorithm}% žüžÄËã·šÇ°×ºÃû³Æ
%\renewcommand{\algorithmicrequire}{\textbf{Input:}}% žüžÄÊäÈëÃû³Æ
%\renewcommand{\algorithmicensure}{\textbf{Output:}}% žüžÄÊä³öÃû³Æ
\footnotesize
\caption{\quad Constructing the complete anchor-set in the reduction}
\label{anchor-set}
\begin{algorithmic}[1]
    \REQUIRE A query sequence $Q$ and  a pan-genome graph $G$;
    %\ENSURE $y = x^n$;
    \STATE $loc(0,Q)\leftarrow \emptyset$; $loc(1,Q)\leftarrow \emptyset$;
    \STATE $loc(0,G)\leftarrow \emptyset$; $loc(1,G)\leftarrow \emptyset$;
    \FOR{$0\leq i< |Q|$}
        %$x\leftarrow Q_{i}$;\\
        \STATE $loc(Q_i,Q)\cup\{Q_i\}$;
    \ENDFOR
    \FOR{$v\in V(G)$}
        %$x\leftarrow \delta(v)$;\\
       \STATE $loc(\delta(v),G)\cup \{(v,0)\}$;
    \ENDFOR
\end{algorithmic}
\end{algorithm}
\fi

\begin{lemma}\label{reduction-complexity}
    The reduction  takes $O(|Q| + |V(G)|)$-time.
\end{lemma}
\begin{proof}
    %The alphabet, query sequence, pan-genome graph are all the same as those in Single-Exa-SGM problem. The set of sequences $R$ is the same as the alphabet.  So the constructions of alphabet, query sequence, pan-genome graph and the se of sequence $R$ take constant time.
    The alphabet, query sequence, pan-genome graph, and the set of sequences $R$ are identical to those defined in the Single-Exa-SGM problem, with $R$ sharing the same composition as the alphabet. Consequently, the construction of these components—alphabet, query sequence, pan-genome graph, and $R$—can be accomplished in constant time.

    %In addition, it takes $O(|Q| + |V(G)|)$-time to obtain the locations. The binary gap-function also takes $O(|Q| + |V(G)|)$-time to be defined. So the reduction needs $O(|Q| + |V(G)|)$-time. 
    Moreover, obtaining the location functions, $loc(q,Q)$ and $loc(q,G)$,  requires $O(|Q| + |V(G)|)$ time, and defining the binary gap-function also takes $O(|Q| + |V(G)|)$ time. Thus, the entire reduction runs in $O(|Q| + |V(G)|)$ time.
\end{proof}

Let $\Pi$ be the set of  instances obtained by the reduction. Let $\omega$ $=$ $|S|$ $-$ $($$Gap_{Q}(S)$ $+$ $Gap_{G}(S)$$)$ be the objective function of the Single-Gap-CLC problem, where $S$ is an anchor-chain. 

\begin{lemma}\label{anchor-nu}
    For each instance in $\Pi$,  if there is an anchor-chain $S$ in $\bigcup_{q\in R}loc(q,Q)\times loc(q,G)$ 
    %in the anchor-set $A$ 
    such that  $\omega\geq |Q|$, then the number of anchors in $S$ is exact $|Q|$.
\end{lemma}
\begin{proof}
    %Since $\omega\geq |Q|$, there are at least $|Q|$  anchors in $S$. Now we prove that the number of anchors in $S$ could not be larger than $|Q|$ . 
    Since $\omega \geq |Q|$, the set $S$ contains at least $|Q|$ anchors. We now demonstrate that the number of anchors in $S$ cannot exceed $|Q|$, establishing the exact count.

    %Suppose on the contrary that $|S|> |Q|$. By the definition, for every two anchors $S_i$ and $S_j$ in $S$, if $i<j$, then $F(S_i) < F(S_j)$. In addition, there are $|Q|$  elements in $Q$, which indicates that there are $|Q|$  locations in $Q$. If $S$ contains more than $|Q|$  anchors, then there are two distinct anchors $S_i$ and $S_j$ in $S$  locating at the same position on $Q$, i.e.,  $F(S_i) = F(S_j)$. This contradicts to the definition of an anchor-chain. Therefore, the number of anchors in $S$ is exact $|Q|$.
    Suppose, to the contrary, that $|S| > |Q|$. By definition, for any two anchors $S_i$ and $S_j$ in $S$ with $i < j$, we have $F(S_i) < F(S_j)$. Since $Q$ contains exactly $|Q|$ elements, there are only $|Q|$ distinct positions available in $Q$. If $S$ had more than $|Q|$ anchors, then by the pigeonhole principle, at least two distinct anchors $S_i$ and $S_j$ in $S$ would occupy the same position in $Q$, implying $F(S_i) = F(S_j)$. This contradicts the strict monotonicity required by the anchor-chain definition. Hence, $|S|$ cannot exceed $|Q|$, and the number of anchors in $S$ must be exactly $|Q|$.
\end{proof}

\begin{lemma}
    \label{one-to-one-corr-Gap}
    For each instance in $\Pi$,  there is an anchor-chain 
   $S$ in $\bigcup_{q\in R}loc(q,Q)\times loc(q,G)$ 
    %in the anchor-set $A$ 
    such that  $\omega\geq |Q|$, if and only if there is a walk $P$ in  $G$ such that the sequence spelled by $P$ is the same as $Q$.
\end{lemma}
\begin{proof}
    Suppose that there is an anchor-chain $S$ such that  $\omega\geq |Q|$. According to Lemma \ref{anchor-nu}, $S$ contains precisely $|Q|$ anchors in $S$. 
    %Recall that $Gap_Q(S)\geq 0$ and $Gap_G(S)\geq 0$ under the restrictions. So $|S| = N_{q}$ must hold true, i.e., every element of $Q$ must be in an anchor of $S$ so as to achieve $\omega\geq N_q(N_q + N_g)$. 
    %This means that all characters of $Q$ are chosen in the resulting solution. Moreover, $Gap_Q(S) = 0$ and $Gap_G(S) = 0$ should hold true simultaneously, which implies that $F(S_{i+1}) = F(S_i) + 1$ for every adjacent anchors $S_i$ and $S_{i+1}$ in $S$.  In addition,  $H(S_{i})\cdot v$ and $H(S_{i+1})\cdot v$ are also adjacent in $G$ for  $0\leq i < |S|-1$.   Therefore, $S$ is a walk in $G$ such that the sequence spelled by $S$ is $Q$.
    This implies that every character of $Q$ is included in the resulting solution. Furthermore, both $Gap_Q(S) = 0$ and $Gap_G(S) = 0$ must hold simultaneously, indicating that $F(S_{i+1}) = F(S_i) + 1$ for any consecutive anchors $S_i$ and $S_{i+1}$ in $S$. Additionally, $H(S_i)\cdot v$ and $H(S_{i+1})\cdot v$ are adjacent in $G$ for all $0 \leq i < |S|-1$. Hence, $S$ forms a walk in $G$ where the sequence induced by $S$ exactly spells $Q$.
    %the same as $Q$. 
    %The construction of the walk $\mathcal{S}$ takes $O(N_q)$ time.

    %Suppose that there is a walk $P$ on  $G$ such that the sequence spelled by $P$ is the same as $Q$. Since every vertex is labeled with a single character, the number of vertices in $P$ is equal to the number of characters in $Q$. 
    Consider a walk $P$ on $G$ whose vertex sequence spells out the same string as $Q$. Since each vertex in $G$ is labeled with a single character, the length of $P$ (in terms of vertices) exactly matches the number of characters in $Q$.
    %Let $P$ $=$ $v_0$, $v_1$, $v_2$, $\dots$, $v_{N_{q}-1}$. We now construct an anchor-chain $S$ $=$ $s_0$, $s_1$, $s_2$, $\dots$, $s_{N_q-1}$, 
    %where $s_{i}$ $=$ $(v_i, [i,i], [0,0])$ 
    %where $F(s_i) = i$ and $H(s_i) = (v_i, 0)$ for ever $s_i:\ 0\leq i\leq N_q-1$. 
    %Let $S$ be an anchor-chain such that $|S| = |P|$, $F(S_i) = i$ and $H(S_i) = (P_i, 0)$ for  $0\leq i\leq |P|-1$.     Then $Gap_Q(S) = 0$, $Gap_G(S) = 0$ and $|S| = |Q|$. Thus we obtain an anchor-chain $S$     $\bigcup_{q\in R}loc(q,Q)\times loc(q,G)$     such that  $\omega\geq |Q|$. 
    Let $S$ be an anchor-chain with $|S| = |P|$, where $F(S_i) = i$ and $H(S_i) = (P_i, 0)$ for $0 \leq i \leq |P| - 1$. Then, $Gap_Q(S) = 0$, $Gap_G(S) = 0$, and $|S| = |Q|$. Hence, we obtain an anchor-chain $S$ in $\bigcup_{q \in R} loc(q, Q) \times loc(q, G)$ such that $\omega \geq |Q|$.
\end{proof}

\begin{theorem}
    \label{Gap-CLC-lower-bound}
    There is no sub-quadratic time algorithm for the Gap-CLC problem, unless the SETH is false, even the alphabet is binary. 
\end{theorem}
\begin{proof}
    %Suppose on the contrary that there is a sub-quadratic time algorithm $\mathcal{F}$ for the Gap-CLC problem. Then $\mathcal{F}$ takes sub-quadratic time  to solve Single-Gap-CLC problem on the set of instances $\Pi$. 
    Suppose, to the contrary, that there exists a sub-quadratic time algorithm $\mathcal{F}$ for the Gap-CLC problem. Then, $\mathcal{F}$ can solve the Single-Gap-CLC problem on the instance set $\Pi$ in sub-quadratic time.
    
    %By Lemma \ref{one-to-one-corr-Gap}, for every instance in $\Pi$, if the objective $\omega\geq |Q|$, then there is a walk on the pan-genome graph such that the sequence spelled by the walk is the same as the given $Q$. On the other hand, if the objective $\omega< |Q|$, then the query sequence could not be  spelled by any walk in the pan-genome graph. In addition, by Lemma \ref{reduction-complexity}, it takes sub-quadratic time to construct the instances in $\Pi$.
    By Lemma \ref{one-to-one-corr-Gap}, for each instance in $\Pi$, whenever the objective value $\omega$ satisfies $\omega \geq |Q|$, there exists a walk in the pan-genome graph whose induced sequence exactly matches the given query sequence $Q$. Conversely, if $\omega < |Q|$, no such walk can spell the entire sequence $Q$. Furthermore, according to Lemma \ref{reduction-complexity}, constructing the instances in $\Pi$ requires sub-quadratic time.
    
    %Consequently,  we have achieved a sub-quadratic time algorithm for the Single-Exa-SGM problem, which contradicts to the recent conclusion that there is no sub-quadratic time algorithm for the Single-Exa-SGM problem, unless the Strong Exponential Time Hypothesis (SETH) is false, even on a binary alphabet \cite{Eaxct-lower-bound}. So  there is no sub-quadratic time algorithm for the Gap-CLC problem, unless the Strong Exponential Time Hypothesis (SETH) is false, even on a binary alphabet. 
    Consequently, we have developed a sub-quadratic time algorithm for the Single-Exa-SGM problem, challenging the recent assertion that no such algorithm exists unless the SETH fails—even over a binary alphabet \cite{Eaxct-lower-bound}. This implies that the Gap-CLC problem also admits no sub-quadratic time algorithm.
\end{proof}

    %In addition, Single-Err-SGM problem has been proven to be NP-hard if errors are allowed on a  graph \cite{original}, even with a binary alphabet \cite{binary-alphabet}. We construct a linear-time reduction from Single-Err-SGM problem to Edit-CLC problem and prove that  the Edit-CLC problem is  NP-hard in the following Theorem \ref{Err-CLC-NP-hard}.
    Moreover, the Single-Err-SGM problem has been shown to be NP-hard when errors are permitted on a graph \cite{original}, even in the case of a binary alphabet \cite{binary-alphabet}. We present a linear-time reduction from the Single-Err-SGM problem to the Edit-CLC problem and establish its NP-hardness in Theorem \ref{Err-CLC-NP-hard}.

\begin{theorem}
    \label{Err-CLC-NP-hard}
    The Edit-CLC problem is NP-hard if  errors are allowed  on  pan-genome graphs, even on a binary alphabet.
\end{theorem}
\begin{proof}
    %{We additionally append the set of sequences is $R = \{0,1\}$ to the instance of Single-Err-SGM problem and let $X_{q} = loc(q,Q)$ and $Y_q = loc(q,G)$ for each $q\in R$. By this way, the Single-Err-SGM problem is reduced to the Single-Edit-CLC problem. }
    We further augment the Single-Err-SGM problem instance by appending the sequence set $R = \{0,1\}$, and define $X_q = \text{loc}(q,Q)$ and $Y_q = \text{loc}(q,G)$ for each $q \in R$. This transformation effectively reduces the Single-Err-SGM problem to the Single-Edit-CLC problem.

    %{In addition, this reduction } can be earned in $O(|Q|+|V(G)|)$-time.  Recall that Amir et al. \cite{original} has proved that the Single-Err-SGM problem is NP-hard if errors are allowed  on the graph, even with a binary alphabet \cite{binary-alphabet}. Therefore, the Edit-CLC problem is NP-hard if  errors are allowed  on pan-genome graphs.
    In addition, this reduction can be computed in $O(|Q|+|V(G)|)$-time. Recall that Amir et al. \cite{original} demonstrated the NP-hardness of the Single-Err-SGM problem when errors are permitted in the graph, even over a binary alphabet \cite{binary-alphabet}. Consequently, the Edit-CLC problem is also NP-hard when errors are allowed in pan-genome graphs.
\end{proof}

\section{4. Conclusions}
\label{conclude}
%Co-linear chaining {problems} have been widely used in sequence-to-sequence alignment. However, it is harder in sequence-to-graph alignment. We studied its complexity in sequence-to-graph alignment and focused on Gap-CLC problem and Edit-CLC problem. The study shows that Gap-CLC problem is not easier than Exa-SGM problem. Moreover, Edit-CLC problem is NP-hard if errors are allowed on pan-genome graphs. Both of these conclusions hold true even on a binary alphabet. This work indicates that in sequence-to-graph alignment models with "co-linear" is at least as complex as those without "co-linear". However, "co-linear" has been used commonly in practice. So it is interesting to close the gap between theory and practice for co-linear chaining in sequence-to-graph alignment. In addition, there are few algorithmic results for the Edit-CLC problem, including exact algorithms and approximation algorithms. A lot of work needs to be done in this direction.  
Co-linear chaining problems have found extensive applications in sequence-to-sequence alignment, yet their extension to sequence-to-graph alignment presents greater challenges. In this work, we investigate the computational complexity of CLC in the context of sequence-to-graph alignment, focusing specifically on the Gap-CLC and Edit-CLC problems. We demonstrate that the Gap-CLC problem is at least as hard as the Exa-SGM problem, while the Edit-CLC problem is proven to be NP-hard when errors are permitted on pan-genome graphs—even over a binary alphabet. These complexity results reveal that incorporating the "co-linear" constraint in sequence-to-graph alignment does not reduce the theoretical complexity compared to non-co-linear models. 

\section{5. Discussions}
\label{discussions}
Despite these theoretical challenges, "co-linear" chaining remains widely adopted in practical tools, highlighting a significant gap between theoretical understanding and real-world practice. Bridging this gap presents an important and intriguing direction for future research. Furthermore, algorithmic solutions for the Edit-CLC problem remain scarce, with limited progress on both exact and approximation algorithms, indicating a clear need for further investigation and development in this area.

%In addition, in our definitions of Gap-CLC problem and Edit-CLC problem, an anchor is not defined as an explicit tuple in this work. Although Lemma \ref{implicit-is-possible} tells us that each explicit anchor-set is  a subset $\bigcup_{q\in R}X_q\times Y_q$, where $R$ is the set of sequences in the anchor-set, $X_q\subseteq loc(q,Q)$ and $Y_q\subseteq loc(q,G)$,  implicit anchor-set would change the input format of co-linear chaining problems.  It is still unknown that whether there is a  linear-time reduction, or a sub-quadratic time reduction,  in the case that an anchor is defined as an explicit tuple. 
Moreover, in this study, we do not define an anchor as an explicit tuple in the formulations of the Gap-CLC and Edit-CLC problems. While Lemma \ref{implicit-is-possible} indicates that every explicit anchor-set is contained within $\bigcup_{q\in R}X_q\times Y_q$, where $R$ denotes the set of sequences in the anchor-set and $X_q\subseteq loc(q,Q)$, $Y_q\subseteq loc(q,G)$, adopting implicit anchor-sets alters the input structure of co-linear chaining problems. It remains an open question whether a linear-time or even a sub-quadratic time reduction exists when anchors are represented as explicit tuples.

{\footnotesize
\iffalse

\fi

\bibliographystyle{plain}
\bibliography{sample}

}

\end{document}